\documentclass[a4paper,UKenglish,cleveref, autoref, thm-restate]{lipics-v2021}
\usepackage{algorithmicx}
\usepackage{amssymb}
\usepackage{algorithm}
\usepackage[noend]{algpseudocode}
\usepackage{tikz}
\usetikzlibrary{arrows.meta, positioning, quotes}
\title{Computational Complexity of Edge Coverage Problem for Constrained Control Flow Graphs} %TODO Please add

\author{Jakub Ruszil}{Jagiellonian University, Faculty of Mathematics and Computer Science, Cracow, Poland }{jakub.ruszil@uj.edu.pl}
{https://orcid.org/0000-0000-0000-0000}{}
\author{Artur Polański}{Jagiellonian University, Faculty of Mathematics and Computer Science, Cracow, Poland }{artur.polanski@uj.edu.pl}
{https://orcid.org/0000-0000-0000-0000}{}
\author{Adam Roman}{Jagiellonian University, Faculty of Mathematics and Computer Science, Cracow, Poland }{adam.roman@uj.edu.pl}
{https://orcid.org/0000-0000-0000-0000}{}
\author{Jakub Zelek}{Jagiellonian University, Faculty of Mathematics and Computer Science, Cracow, Poland }{jakub.zelek@doctoral.uj.edu.pl}
{https://orcid.org/0000-0000-0000-0000}{}
\authorrunning{xxx} %TODO mandatory. First: Use abbreviated first/middle names. Second (only in severe cases): Use first author plus 'et al.'

\Copyright{xxx} %TODO mandatory, please use full first names. LIPIcs license is "CC-BY";  http://creativecommons.org/licenses/by/3.0/


 %TODO mandatory; please add comma-separated list of keywords

\category{} %optional, e.g. invited paper

\relatedversion{} %optional, e.g. full version hosted on arXiv, HAL, or other respository/website
\EventEditors{John Q. Open and Joan R. Access}
\EventNoEds{2}
\EventLongTitle{42nd Conference on Very Important Topics (CVIT 2016)}
\EventShortTitle{CVIT 2016}
\EventAcronym{CVIT}
\EventYear{2016}
\EventDate{December 24--27, 2016}
\EventLocation{Little Whinging, United Kingdom}
\EventLogo{}
\SeriesVolume{42}
\ArticleNo{23}
\begin{document}

\maketitle

\begin{abstract}
The article studies edge coverage for control flow graphs extended with explicit constraints. Achieving a given level of white-box coverage for a given code is a classic problem in software testing. We focus on designing test sets that achieve edge coverage \textit{while respecting additional constraints} between vertices. The paper analyzes how such constraints affect both the feasibility and computational complexity of edge coverage.

The motivation for taking constraints into account stems from limitations of plain control flow graphs, known as the \textit{semantic gap}. Standard graphs admit paths that represent infeasible executions, violate required execution patterns, or exceed realistic execution costs. Without constraints, edge coverage may require test paths that no execution satisfies or that violate domain rules. Constraints model semantic relations, mandatory execution patterns, and resource limits directly at the graph level. This leads to test sets that better reflect real program behavior and testing goals.

The paper discusses five types of constraints. POSITIVE constraints require at least one test path where a given vertex precedes another. NEGATIVE constraints forbid any such test path. ONCE constraints require exactly one test path with a single occurrence of one vertex before another. MAX ONCE constraints allow such precedence in at most one test path. ALWAYS constraints require every test path containing a given vertex to also contain another vertex later on the same path. Each type models a different test requirement, such as mandatory flows, semantic exclusions, or execution cost limits.

We investigate the computational complexity of finding a test set that achieves edge coverage and respects a given set of constraints. For POSITIVE constraints, the existence of an edge covering test set is decidable in polynomial time by extending standard edge coverage constructions with additional paths for each constraint. For NEGATIVE, MAX ONCE, ONCE, and ALWAYS constraints, the decision problem is NP-complete. The proofs rely on polynomial reductions from variants of SAT. The NP-completeness results hold even for restricted graph classes, including acyclic graphs, for all these four constraints.

Finally, we study the fixed-parameter tractability of the NEGATIVE constraint. Although the general problem is NP-complete, the paper presents an FPT algorithm with respect to the number of constraints. 
\end{abstract}

\newpage
\section{Introduction}

% rola technik bialoskrzynkowych i ich zalety; modelowanie programow przez grafy skierowane

White-box testing leverages knowledge of a program’s internal structure to design and evaluate test cases. Unlike black-box methods that only validate observable behavior, white-box techniques aim to directly exercise the code’s control and data flows, thereby enabling early detection of latent defects and vulnerabilities. Central to this approach are coverage criteria -- quantifiable benchmarks which measure the extent to which test sets exercise the program’s control-flow graph (CFG), helping testers identify untested areas and improve test adequacy.

% kryteria pokrycia w literaturze i uzywane w praktyce (rola edge coverage)
There are several well-known white-box coverage criteria, such as statement, branch, MC/DC, and N-switch coverage \cite{rechtberger2022a}. Path testing \cite{Sun2026} is another approach that defines coverage items as paths through the CFG. It shifts the focus of testing from isolated program elements to meaningful execution sequences. Real failures often arise not from a single line of code, but from the interaction of decisions over time -- exactly what paths capture. Many approaches for designing test sets that satisfy path-based coverage from CFGs include simple path coverage \cite{ammann2016}, basis path coverage \cite{Ghiduk2014, yan2008, Sun2024}, NPATH coverage \cite{ntafos1979}, or prime path coverage \cite{ammann2016, li2009}. 

% uzasadnienie dlaczego constrainty sa potrzebne w modelowaniu
However, in practice, modeling programs solely as CFGs provides an overly permissive foundation for path-based coverage criteria, as it admits many test paths that are neither semantically feasible nor economically meaningful to explore. First, CFGs suffer from a so-called \textit{semantic gap}: syntactic reachability over-approximates true program behavior, leading to the inclusion of infeasible control-flow paths (i.e., paths that cannot be exercised under any concrete execution \cite{yan2008, Gotlieb2010, Jaffar2010}). Second, CFGs fail to encode obligatory execution patterns that arise from domain-specific control flows or mandated use-case scenarios, where certain sequences of nodes must occur together even though this relationship is not imposed directly by the CFG \cite{He2024, Gong2024}. Third, CFGs ignore execution-cost constraints \cite{Nikolik2012, Wu2016, Rabatul2023, Polanski2025}, meaning that they may include paths whose feasibility is limited not by semantics but by prohibitive runtime or resource consumption, and which therefore can be traversed only a bounded number of times in practice. 

% wniosek: zaproponowano constraint path-based testing; przeglad literatury - prace Buresa i inne dot. pokrycia, path-based testing, etc.
As a result, path-based coverage criteria grounded purely in CFG structure risk both overestimate achievable coverage and misallocate testing effort toward paths that are either impossible, irrelevant, or economically unjustifiable. Therefore, a \textit{constraint path-based testing} approach was proposed \cite{Klima2024}. In \cite{Klima2025a}, four types of constraints were defined to model code-related semantic aspects in a CFG. In \cite{Klima2025c}, an algorithm was proposed for designing test sets achieving edge coverage (EC) respecting the defined constraints.

Let us consider a practical, motivational example involving constrained path-based testing of the Enterprise Procurement and Vendor Management business process shown in Fig. \ref{business-process}.

\begin{figure}[!ht]
\centering
\includegraphics[scale=0.5]{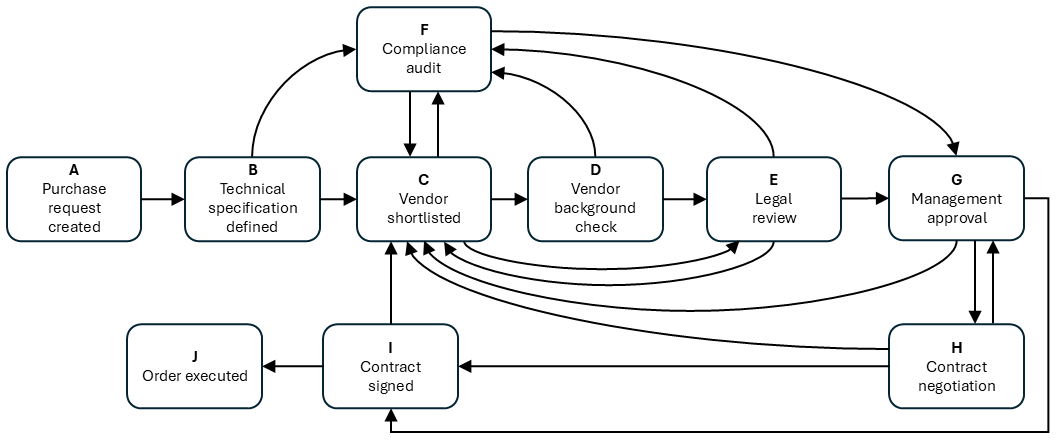}
\caption{Enterprise procurement and vendor management business process model}
\label{business-process}
\end{figure}

The graph is intentionally permissive: loops are allowed; after E, F, or G, the process may return to C (vendor changed); after H, negotiation may fail and return to C; audits and reviews can occur multiple times structurally; approval does not terminate the process; signing does not prevent further checks structurally.

However, for practical reasons, we want to introduce certain restrictions on the form of test paths in this model, as shown below. These five restrictions present five different constraint types defined in Section \ref{sec:problem}:

\begin{itemize}
    \item ($C_+$) After B, there should be at least one test that exercises F: some specifications (B) trigger regulatory constraints, so auditing (F) must be validated at least once.
     \item ($C_-$) After I, exercising F is forbidden: once a contract is signed (I), a compliance audit (F) must never occur afterward in any test case. Compliance audits are pre-contract controls, and post-signature audits would invalidate the contract.
     \item ($C_1$) Exactly one test case in the entire test suite may contain a path where F is followed by H: Compliance audit (F) together with negotiation (H) is rare, very expensive to simulate, and often involves external legal/compliance tools. Testing it zero times means we miss a critical scenario, but testing it more than once is a wasted cost.
     \item ($C_{\leq 1}$) Background check (D) followed by legal review (E) may occur in at most one test case: background with legal review is a high-cost combination, so repeating it gives diminishing returns. However, other paths (e.g., D-F or C-E directly) provide reasonable and frequent business flows, so in case we are short on time, we may skip a path D-E in our tests.
     \item ($C_=$) After H, G must follows: in any test case, if negotiation (H) occurs, approval (G) must occur later in the same test case. Negotiated contracts must be explicitly approved. Auto-signing is only allowed for non-negotiated templates.
\end{itemize}

Notice that the following set of three test cases achieves EC and respects all above-mentioned five constraints:

\begin{enumerate}
    \item ABFCFCDFGCFGIJ
    \item ABCECEGHGHCEGHICDEGIJ
    \item ABCEFGHGIJ
\end{enumerate}

For many white-box coverage criteria, it is known whether the problem of finding an optimal test set (in terms of different cost criteria, such as total number of edges, minimal number of test paths, etc.) is polynomial or not (cf. \cite{Polanski2025}). However, there is no theoretical research on the complexity of constraint path-based testing, even in the simplest case of EC, which is a CFG equivalent of branch coverage.

%Cel i struktura naszej pracy
In this article, we address this issue by showing that, for most types of constraints, the problem of finding a test set that achieves EC and respects a given set of constraints becomes NP-complete. The structure of the paper is as follows. In Section \ref{sec:problem}, we introduce all the necessary definitions and formally define the problem. Section 3 contains the main results on the computational complexity of the above-mentioned problem for five different constraint types. In Section 4, we show that, for the case ``NEGATIVE'', the problem is FPT with respect to the number of constraints. Section 5 follows with conclusions.

\section{Problem statement} \label{sec:problem}

Let $H=(V, E)$ be a directed graph. For $v \in V$ we define $d_{in}(v) = |(V \times \{v\}) \cap E|$ and $d_{out}(v) = |(\{v\} \times V) \cap E|$. A sequence $(v_1, v_2, \ldots, v_n)$, $n \geq 1$, $v_i \in V, i \in \{1, \ldots, n\}$ such that $(v_i, v_{i+1}) \in E$ for $i \in \{1, \ldots, n-1\}$ is called a \textit{path} in $H$. A set of all paths in $H$ is denoted by $P(H)$. The \textit{length} $|p|$ of a path $p=(v_1, \ldots,v_n)$ is $n-1$, the number of its edges. If $p = (v_1, \ldots, v_m)$ and $q = (v_{m+1}, \ldots, v_n)$ are paths in $P(H)$, and $(v_m, v_{m+1}) \in E$, then $p.q = (v_1, \ldots, v_n)$ is a path in $P(H)$, and we call it a  \textit{concatenation} of  $p$ and $q$. By $\#_vp$ we denote the number of occurrences of $v$ in $p$. A vertex $w \in V$ is \textit{reachable from} $v \in V$ (and we denote it by $v \rightsquigarrow w$) if there exists a path $(v, \ldots, w)$ in $H$. Each path $(v_1, \ldots, v_n)$ for $n \geq 2$ can be represented equivalently as a sequence of edges, $(e_1, \ldots, e_{n-1})$, where $e_i=(v_i, v_{i+1})$, $i \in \{1, \ldots, n-1\}$.

A \textit{constrained Control Flow Graph} (cCFG) is a 6-tuple $G=(V, E, s, t, C, X)$, where $(V, E)$ is a directed graph, $|V| \geq 2$, $s, t \in V$, $s \neq t$, $d_{in}(s)=0$, $d_{out}(t)=0$, $C \subset V \times V$ is a set of constraints, $X \in \{C_+, C_-, C_1, C_{\leq 1}, C_=\}$ is a constraint type, and $\forall v \in V\ s \rightsquigarrow v \wedge v \rightsquigarrow t$.

Let $G=(V, E, s, t, C, X)$ be a cCFG and let $p=(v_1, \ldots, v_n)$, $n \geq 2$, be a path in $G$. If $v_1=s$ and $v_n=t$ we call $p$ a \textit{test path}. By $TP_G$ we denote a set of all test paths in $G$. A \textit{test set} $TS_G$ for a cCFG $G$ is a set of test paths, i.e., $TS_G \subset TP_G$. We say that a test set $TS_G$ is \textit{admissible} if it respects the constraints in the following way:

\begin{itemize}
\item if $X=C_+$ then $\forall (x, y) \in C \ \exists p_1,p_2,p_3 \in P(G) \ \exists p \in TS_G\ p = p_1.x.p_2.y.p_3$,
\item if $X=C_-$ then $\forall (x, y) \in C\ \forall p_1, p_2, p_3 \in P(G)\ \forall p \in TS_G\ p \neq p_1.x.p_2.y.p_3$,
\item if $X=C_1$ then $\forall (x, y) \in C\ \exists p_1,p_2,p_3 \in P(G) \ \exists p \in TS_G\ p = p_1.x.p_2.y.p_3$, $\#_xp_1=\#_xp_2=0$, $\#_yp_2=\#_yp_3=0$, and ($\forall q \in TS_G\ q = p_1'.x.p_2'.y.p_3'$ for some $p_1', p_2', p_3' \in P(G)) \Rightarrow (q = p)$,
\item if $X=C_{\leq 1}$ then $\forall (x, y) \in C\ |A| \leq 1 \wedge \forall q \in TS_G\ (q = p_1'.x.p_2'.y.p_3'$ for some $p_1',p_2',p_3' \in P(G)) \Rightarrow q \in A$, where $A = \{p \in TS_G\ p = p_1.x.p_2.y.p_3$ for some $p_1,p_2,p_3 \in P(G)$ such that $\#_xp_1=\#_xy=p_2$, $\#_yp_2=\#_yp_3=0\}$,
\item if $X=C_=$ then $\forall (x, y) \in C\ \forall p \in TS_G\ (p=p_1.x.p_2 $ for some $p_1,p_2 \in P(G))\Rightarrow (p_2 = p_3.y.p_4$ for some $p_3,p_4\in P(G))$.
\end{itemize}

%\begin{itemize}
%    \item if $X=C_+$ then $\forall (v, w) \in C\ |{(s, \ldots, v, \ldots, w, \ldots, t) \in TS_G}| \geq 1$,
%    \item if $X=C_-$ then $\forall (v, w) \in C\ |{(s, \ldots, v, \ldots, w, \ldots, t) \in TS_G}| = 0$,
%    \item if $X=C_1$ then $\forall (v, w) \in C\ |{(s, \ldots, v, \ldots, w, \ldots, t) \in TS_G}| = 1$
%    \item if $X=C_{\leq 1}$ then $\forall (v, w) \in C\ |{(s, \ldots, v, \ldots, w, \ldots, t) \in TS_G}| \leq 1$
%    \item if $X=C_=$ then $\forall (v, w) \in C\ |{(s, \ldots, v, \ldots, t) \in TS_G}| =  |{(s, \ldots, v, \ldots, w , \ldots, t) \in TS_G}|$
%\end{itemize}

The first four types of constraints ($C_+, C_-, C_1, C_{\leq 1}$) were introduced in \cite{Klima2025a}. A $C_+$ (``POSITIVE'') constraint $(x, y)$ requires that $x$ precedes $y$ in at least one test path in $TS_G$. A $C_-$ (``NEGATIVE’‘) constraint $(x, y)$  requires that there is no test path in $TS_G$ with $x$ preceeding $y$. A $C_1$ (``ONCE'') constraint $(x, y)$ requires that there is exactly one test path in $TS_G$ in which $x$ precedes $y$; moreover, there can be only one $x$ before $y$ and only one $y$ after $x$. A $C_{\leq 1}$ (``MAX-ONCE’‘) constraint $(x, y)$ is similar to ``ONCE'' but requires that $x$ precedes $y$ in at most one test path in $TS_G$. Additionally, we introduce a new $C_=$ (``ALWAYS’') constraint $(x, y)$ requiring that if $x$ is on a test path, then it must be succeeded by $y$ in that test path. $C_+$ constraints model important control flows that must be tested, $C_-$ and $C_=$ constraints model semantic gaps in a source code, while $C_1$ and $C_{\leq 1}$ model execution cost-type constraints. 

Let $G=(V, E, s, t,C,X)$ be a cCFG, $TS_G$ be a test set for $G$, $e \in E$, and $p = (v_1, \ldots, v_n) \in TS_G$, $n\geq 2$. We say that $p$ \textit{covers} $e$, and write $e \in p$, if there exists $i \in \{1, \ldots, n-1\}$ such that $e=(v_i, v_{i+1})$. We say that $TS_G$ \textit{achieves EC for $G$ respecting $C$ of type $X$} if $TS_G$ is admissible and $\forall e \in E\  \exists p \in TS_G: e \in p.$

The main problem we investigate is as follows. Given a cCFG $G=(V, E, s, t, C, X)$, what is the computational complexity of finding a test set that achieves $EC$ for $G$ respecting $C$ of type $X$?

\section{Theoretical results on problem complexity}
It is easy to see that deciding whether there exists a test set that achieves EC for a given cCFG $G=(V, E, s, t, C, X)$ respecting $C$ of type $X$ belongs to NP. Every solution $TS_G$ can be verified in polynomial time by a deterministic Turing machine as follows. First, for each $e \in E$, we check if there exists $p \in TS_G$ such that $e \in p$. Next, for each constraint, we check whether it is respected by $TS_G$.  
The next subsections present more specialized results in that matter.

\subsection{POSITIVE}

Consider $G=(V, E, s, t, C, C_+)$, a cCFG with a set of POSITIVE constraints. It was proven in \cite{Polanski2025} that constructing a test set $TS$ achieving $EC$ for a CFG without a set of constraints is possible in polynomial time. We can construct $TS_G$ from it that achieves $EC$ while also respecting $C$ (provided it is possible). In order to do so, it suffices to check (for example using DFS or BFS) that for each constraint $(x,y)\in C$  $x$ is reachable from $s$, $y$ is reachable from $x$ and $t$ is reachable from $y$, getting a path of the form $(s,\ldots, x,\ldots, y, \ldots, t)$. After adding such a path for every constraint in $C$ to a set achieving $EC$ for $G$, we get a test set respecting $C$ that still achieves $EC$, in polynomial time. In case when at least one of the constraints was impossible to satisfy, we also get that information in polynomial time.

\subsection{NEGATIVE}

The following theorem, along with its proof, is very similar to Theorem 2 from \cite{Polanski2025}. The main difference is that in \cite{Polanski2025} constraints are edge-based rather than vertex-based as considered here. For completeness, we present the full proof rather than listing only the necessary changes.

\begin{theorem}\label{th_npc_negative}

Let $G=(V, E, s, t, C, C_-)$ be a cCFG with a set of NEGATIVE constraints. Deciding whether a test set achieving EC for $G$ respecting $C$ of type $C_-$ exists is NP--complete.
\end{theorem}

\begin{proof}
We will reduce 3--SAT to the problem above. Let $n,k$ be positive integers and $\varphi = \varphi_1 \wedge \ldots \wedge \varphi_n$ be an instance of the 3--SAT problem over a set of variables $X=\{x_1,\ldots,x_k\}$. Each clause $\varphi_i$, $i = 1, \ldots, n$, is of the form $x_{i,1} \vee x_{i,2} \vee x_{i,3}$ with each $x_{i,1}, x_{i,2}, x_{i,3}$ being either an instance of a variable from $X$ or its negation. The idea is as follows. Each $\varphi_i$ corresponds to a gadget like the one depicted in Fig. \ref{gadgets}(a). The gadgets are joined through in-between nodes, by edges depicted in blue (bold arrows) and red (dashed arrows), as seen in Fig. \ref{gadgets}(b), with constraints that forbid using a red node after a blue one and vice versa. Another set of constraints is added that forbids using nodes corresponding to an instance of a variable after exercising the node corresponding to the negation of said variable, and vice versa.

%\begin{figure}[!ht]
%\centering
%\includegraphics[scale=0.2]%{fig_singleGadget}
%\caption{A single gadget representing a clause $\varphi_i$}
%\label{figGadget}
%\end{figure}

Formally, let $G=(V, E, s, t, C, C_-)$ be a cCFG such that:

\begin{itemize}
    \item $V=V_{black}\cup V_{blue}\cup V_{red}$,
    \item $V_{black} = \{s, t\} \cup \{v_{i, 1}, v_{i, 2}, v_{i, 3}, v_i, w_i\ |\ i \in \{1, \ldots n\}\}$,
    %\item $V_{black}=\{s,t\}\cup\{v_{1,1},\ldots,v_{n,1}\}\cup\{v_{1,2},\ldots,v_{n,2}\}\cup\{v_{1,3},\ldots,v_{n,3}\}\cup$\\
    %$\cup\{v_1,\ldots,v_n\}\cup\{w_1,\ldots,w_n\}$,
    \item $V_{blue}=\{b_1,\ldots,b_n,b_{n+1}\}$,
    \item $V_{red} = \{p_i, q_i, y_i, z_i\ |\ i \in \{1, \ldots, n\}\}$,
    %\item $V_{red}=\{p_1,\ldots,p_n\}\cup\{q_1,\ldots,q_n\}\cup\{y_1,\ldots,y_n\}\cup\{z_1,\ldots,z_n\}$,
    \item $E=E_{black} \cup E_{blue} \cup E_{red}$,
    \item $E_{black}=\{(v_i, v_{i, j}), (v_{i, j}, w_i)\ |\ i \in \{1, \ldots, n\}, j \in \{1, 2, 3\}\}$,
    %\item $E_{black}=\{(v_i,v_{i,1}) | i \in \{1,\ldots,n\}\}\cup\{(v_i,v_{i,2}) | i \in \{1,\ldots,n\}\}\cup$\\
    %$\cup\{(v_i,v_{i,3}) | i \in \{1,\ldots,n\}\}\cup\{(v_{i,1},w_i) | i \in \{1,\ldots,n\}\}\cup$\\
    %$\cup\{(v_{i,2},w_i) | i \in \{1,\ldots,n\}\}\cup\{(v_{i,3},w_i) | i \in \{1,\ldots,n\}\}$,
    \item $E_{blue} = \{(s, b_1), (b_{n+1}, t)\} \cup \{(w_i, b_{i+1}), (b_i, v_i)\ |\ i \in \{1, ..., n\}\},$
    %\item $E_{blue}=\{(s,b_1)\}\cup\{(w_i,b_{i+1}) | i \in \{1,\ldots,n\}\}\cup$\\
    %$\cup\{(b_i,v_i) | i \in \{1,\ldots,n\}\}\cup\{(b_{n+1},t)\}$,
    \item $E_{red} = \{(s, p_i), (s, q_i), (p_i, v_i), (q_i, v_i), (w_i, y_i), (w_i, z_i), (y_i, t), (z_i, t)\ |\ i \in \{1, ..., n\}\}$.
    %\item $E_{red}=\{(s,p_i) | i \in \{1,\ldots,n\}\} \cup \{(s,q_i) | i \in \{1,\ldots,n\}\} \cup$\\
    %$\cup\{(p_i,v_i) | i \in \{1,\ldots,n\}\}\cup\{(q_i,v_i) | i \in \{1,\ldots,n\}\} \cup$\\
    %$\cup\{(w_i,y_i) | i \in \{1,\ldots,n\}\}\cup\{(w_i,z_i) | i \in \{1,\ldots,n\}\} \cup$\\
    %$\cup\{(y_i,t) | i \in \{1,\ldots,n\}\}\cup\{(z_i,t) | i \in \{1,\ldots,n\}\}$,
\end{itemize}

\begin{figure}[!ht]
\centering
\includegraphics[scale=0.4]{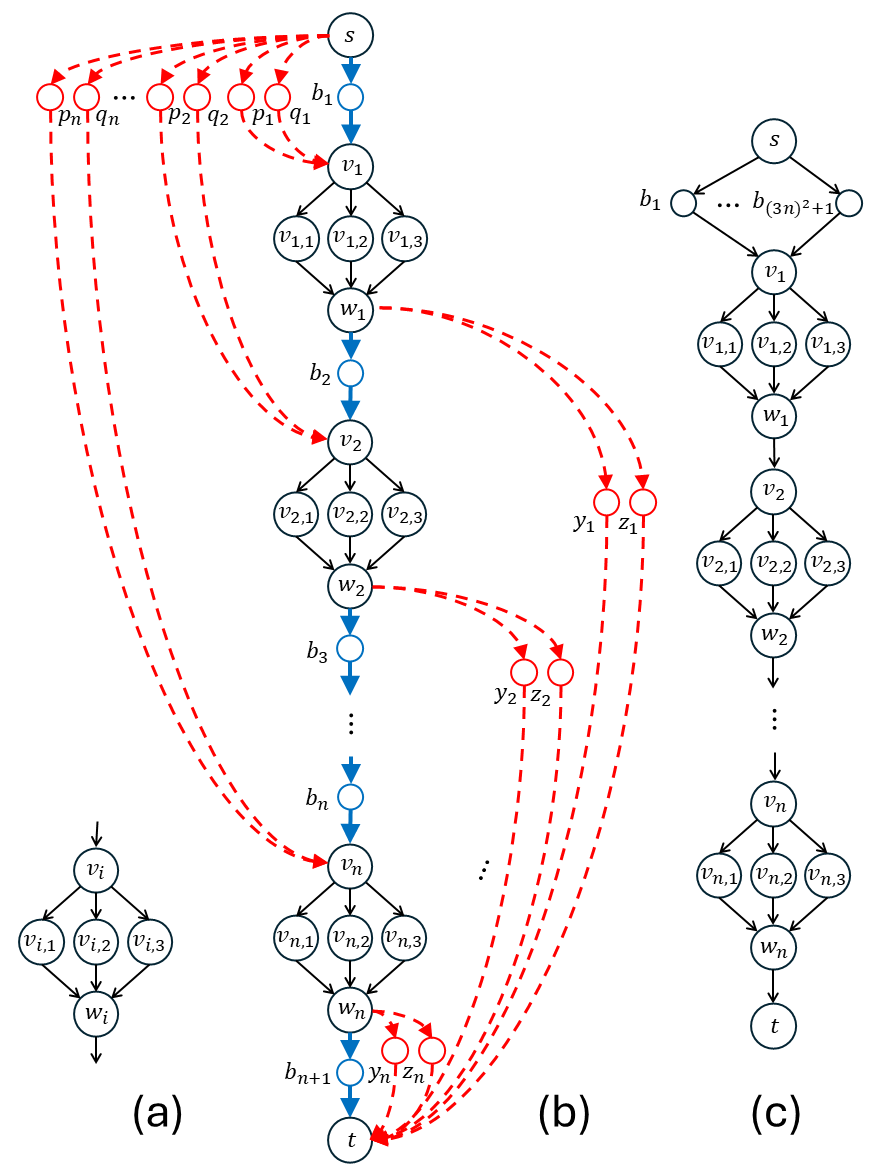}
\caption{Three gadgets used in Theorems \ref{th_npc_negative}, \ref{th_npc_maxonce}, and \ref{th_npc_once}}
\label{gadgets}
\end{figure}

It remains to define a suitable set of constraints for $G$. Let 
\begin{eqnarray*}
C_1 &=& \{(r,b)\ |\ r\in V_{red}, b \in V_{blue}\} \cup \{(b,r)\ |\ b\in V_{blue}, r \in V_{red}\}, \\
C_2 & = &  \{(v_{i,j}, v_{i',j'})\ |\ i, i' \in \{1, \ldots, n\},\ j,j'\in\{1,2,3\}, \\
& & x_{i,j}=\neg x_{i',j'} \mbox{\ (as variables in } \varphi) \}
\end{eqnarray*}

%$$C_1 = \{(r,b)\ |\ r\in V_{red}, b \in V_{blue}\} \cup \{(b,r)\ |\ b\in V_{blue}, r \in V_{red}\},$$ 
%$$\begin{array}{lll} C_2 & = &  \{(v_{i,j}, v_{i',j'})\ |\ i, i' \in \{1, \ldots, n\},\ j,j'\in\{1,2,3\},\ \\ & & x_{i,j},x_{i',j'} \mbox{\ correspond to a variable and its negation in } \varphi \}.\end{array}$$ 

Finally, let $C=C_1 \cup C_2$ be a set of constraints for $G$. %Note that if a test set $T$ achieves EC for $G$ respecting $C$, then $|T| \geq 2n+1$, with a total cost of at least $|E|=15n+1$.

We will now prove that if $\varphi$ is satisfiable, then there exists a test set $T_\varphi$ achieving EC for $G$ respecting $C$. Fix $\nu: X \rightarrow \{0, 1\}$, an assignment of logical values to the variables, that satisfies $\varphi$. For the first element of $T_\varphi$, choose a path $t_0$ using only vertices from $V_{blue}$ and $V_{black}$  corresponding to $\nu$ in the following way. At each $v_i$, $i \in \{1, \ldots, n\}$, choose an edge leading to a vertex that corresponds to either a variable that is set as true by $\nu$ or a negation of a variable that is set as false by $\nu$ (one of these options must be possible at each $i$, since each $\varphi_i$ must be satisfied). We are left with exactly $2n$ edges originating in $s$, so we need at least $2n$ more paths in $T_\varphi$. Since for any $i \in \{1, \ldots, n\}$ we have exactly two paths from $s$ to $v_i$ that do not use edges from $E_{blue}$ (namely $(s,p_i,v_i)$ and $(s,q_i,v_i)$), as well as exactly two paths from $w_i$ to $t$ that do not use edges from $E_{blue}$ (namely $(w_i,y_i,t)$ and $(w_i,z_i,t)$), for each $i$ we can add to $T_\varphi$ two paths that do not have common edges and together with $t_0$ cover a gadget representing $\varphi_i$, with each edge of it covered exactly once. Since we need two such paths for each gadget, together with $t_0$, we have $2n+1$ paths in $T_\varphi$. Note that each edge in $G$ was used exactly once by $T_\varphi$ and that $T_\varphi$ respects $C$.

We will now prove that if there exists a test set $T_\varphi$ achieving EC for $G$ respecting $C$, then $\varphi$ is satisfiable. Since $T_\varphi$ respects $C$, any test path starting with the edge $(s,b_1)$ must only use vertices from $V_{black}$ and $V_{blue}$. Further, since $T_\varphi$ achieves EC for $G$, there must exist at least one such path $p$. Denote the sequence of numbers from 1 to 3 that correspond to the path $p$ taken in each gadget corresponding to each $\varphi_i$ by $(k_1,\ldots,k_n)$. Define $\nu: X \rightarrow \{0, 1\}$ by setting values for the variables in $X$ such that each $x_{i,k_i}$ is true (note that it does not lead to contradictions thanks to $C_2$). For any remaining variable $x$, set $\nu(x)$ to any value. It is easy to see that $\nu$ satisfies $\varphi$.
\end{proof}

\subsection{MAX-ONCE}

\begin{theorem}\label{th_npc_maxonce}
Let $G=(V, E, s, t, C, C_{\leq 1})$ be a cCFG with a set of MAX-ONCE constraints. Deciding whether a test set achieving EC for $G$ respecting $C$ of type $C_{\leq 1}$ exists is NP--complete.
\end{theorem}

\begin{proof}
We will reduce 3--SAT to the problem above. Let $n,k$ be positive integers and $\varphi = \varphi_1 \wedge \ldots \wedge \varphi_n$ be an instance of the 3--SAT problem over a set of variables $X=\{x_1,\ldots,x_k\}$. Each clause $\varphi_i$, $i = 1, \ldots, n$, is of the form $x_{i,1} \vee x_{i,2} \vee x_{i,3}$ with each $x_{i,1}, x_{i,2}, x_{i,3}$ being either an instance of a variable from $X$ or its negation. The idea is as follows. Each $\varphi_i$ will correspond to a gadget like the one depicted in Fig. \ref{gadgets}(a). The gadgets will be joined by edges in sequence, as seen in Fig. \ref{gadgets}(c), with constraints that restrict using nodes corresponding to an instance of a variable after using the node corresponding to the negation of said variable and vice versa.

%\begin{figure}[!ht]
%\centering
%\includegraphics[scale=0.2]{fig_singleGadget2}
%\caption{A single gadget representing a clause $\varphi_i$}
%\label{figGadget2}
%\end{figure}

Formally, let $G=(V, E, s, t, C, C_{\leq 1})$ be a cCFG such that:

\begin{itemize}
    \item $V=\{s, t, b_1, \ldots, b_{(3n)^2+1}\} \cup \{v_{i, 1}, v_{i, 2}, v_{i, 3}, v_i, w_i\ |\ i \in \{1, \ldots, n\}\}$,
    %\item $V=\{s,t\}\cup\{v_{1,1},\ldots,v_{n,1}\}\cup\{v_{1,2},\ldots,v_{n,2}\}\cup\{v_{1,3},\ldots,v_{n,3}\}\cup$\\
    %$\cup\{v_1,\ldots,v_n\}\cup\{w_1,\ldots,w_n\}\cup\{b_1,\ldots,b_{(3n)^2+1}\}$.
    \item $E=\{(v_i, v_{i, j}), (v_{i, j}, w_i)\ |\ i \in \{1,\ldots,n\}, j \in \{1, 2, 3\}\} \cup \{(w_i, v_{i+1})\ |\ i \in \{1,\ldots,n-1\}\}\\$
    $\cup \{(s, b_i), (b_i, v_1)\ |\ i \in \{1,\ldots,(3n)^2+1\}\}  \cup \{(w_n, t)\}$.
    %\item $E=\{(v_i,v_{i,1}) | i \in \{1,\ldots,n\}\}\cup\{(v_i,v_{i,2}) | i \in \{1,\ldots,n\}\}\cup$\\
    %$\cup\{(v_i,v_{i,3}) | i \in \{1,\ldots,n\}\}\cup\{(v_{i,1},w_i) | i \in \{1,\ldots,n\}\}\cup$\\
    %$\cup\{(v_{i,2},w_i) | i \in \{1,\ldots,n\}\}\cup\{(v_{i,3},w_i) | i \in \{1,\ldots,n\}\}\cup$\\
    %$\cup\{(s,b_i) | i \in \{1,\ldots,(3n)^2+1\}\}\cup\{(b_i,v_1) | i \in \{1,\ldots,(3n)^2+1\}\}\cup$\\
    %$\cup\{(w_i,v_{i+1}) | i \in \{1,\ldots,n-1\}\}\cup\{(w_n,t)\}$.
\end{itemize}

%\begin{figure}[!ht]
%\centering
%\includegraphics[scale=0.3]{fig_fullGadget2}
%\caption{A full gadget corresponding to a formula $\varphi$}
%\label{figGadgetsAll2}
%\end{figure}

It remains to define a suitable set of constraints for $G$. Let $$\begin{array}{lll} C & = &  \{(v_{i_1,n_1}, v_{i_2,n_2})\ |\ i_1<i_2,\ i_1, i_2 \in \{1, \ldots, n\},\ n_1,n_2\in\{1,2,3\},\ \\ & & x_{i_1,n_1}=\neg x_{i_2,n_2} \mbox{\ (as variables in } \varphi) \}.\end{array}$$

We will now prove that if $\varphi$ is satisfiable, then there exists a test set $T_\varphi$ achieving EC for $G$ respecting $C$. Fix $\nu: X \rightarrow \{0, 1\}$, an assignment of logical values to the variables that satisfies $\varphi$. For the first element of $T_\varphi$, choose a path $t_1$ starting with vertices $s$ and $b_1$, corresponding to $\nu$ in the following way. At each $v_i$, $i \in \{1, \ldots, n\}$, choose an edge leading to a vertex that corresponds to either a variable that is set as true by $\nu$ or a negation of a variable that is set as false by $\nu$ (one of these options must be possible at each $i$, since each $\varphi_i$ must be satisfied). Now add to $T_\varphi$ paths $t_2$ through $t_{(3n)^2+1}$ copying $t_1$ almost exactly, differing only on a second vertex where instead of $b_1$ we use $b_i$, $i \in \{2,\ldots,(3n)^2\}$. Note that none of those paths contains a sequence of nodes that is restricted by $C$, since we used $\nu$ for their construction. We are left with exactly one edge originating in $s$. For $i\in \{1, 2, 3\}$ we define $t_i'$ by
$$t_i' = (s,b_{(3n)^2+1},v_1,v_{1,i},w_1,v_2, v_{2, i}, w_2,\ldots,v_n,v_{n,i},w_n,t).$$

Note that each pair of vertices from $C$ is contained at most once in at most one of the paths $t_1'$, $t_2'$, $t_3'$, and that each gadget corresponding to $\varphi_i$ has its edges covered by those three paths alone. Adding $t_1'$, $t_2'$, $t_3'$ to $T_\varphi$, we get a test set achieving EC for $G$ respecting $C$.

We will now prove that if there exists a test set $T_\varphi$ achieving EC for $G$ respecting $C$, then $\varphi$ is satisfiable. Since there are $(3n)^2+1$ edges outgoing from $s$, $T_\varphi$ must contain at least as many paths. Further, since there are $3n$ vertices that correspond to a variable or its negation (i.e. those of the form $v_{l,m}$, $l\in \{1, \ldots, n\}$, $m\in \{1, \ldots, 3\}$), there are at most $(3n)^2$ pairs in $C$. Therefore, since there are more paths in $T_\varphi$ than the number of pairs restricted by $C$ (with each pair appearing in at most one path from $T_\varphi$), there must exist at least one path, let us denote it by $t_0$, which is in $T_\varphi$ but does not contain any pairs of vertices restricted by $C$. Denote the sequence of numbers from 1 to 3 that correspond to the path $t_0$ taken in each gadget corresponding to each $\varphi_i$ by $(k_1,\ldots,k_n)$. Define $\nu: X \rightarrow \{0, 1\}$ by setting values for the variables in $X$ such that each $x_{i,k_i}$ is true (note that it does not lead to contradictions thanks to $t_0$ not containing \textit{any} pairs of vertices restricted by $C$). For any remaining variable $x$, set $\nu(x)$ to any value. It is easy to see that $\nu$ satisfies $\varphi$.
\end{proof}

\subsection{ONCE}

\begin{theorem}\label{th_npc_once}
Let $G=(V, E, s, t, C, C_1)$ be a cCFG with a set of ONCE constraints. Deciding whether a test set achieving EC for $G$ respecting $C$ of type $C_1$ exists is NP--complete.
\end{theorem}

\begin{proof}
The idea of the proof is to reduce a version of 3--SAT, where each variable or its negation may appear in a clause no more than once, to the problem above. Let $n,k$ be positive integers and $\varphi = \varphi_1 \wedge \ldots \wedge \varphi_n$ be an instance of the 3--SAT problem over a set of variables $X=\{x_1,\ldots,x_k\}$. Each clause $\varphi_i$, $i = 1, \ldots, n$, is of the form $x_{i,1} \vee x_{i,2} \vee x_{i,3}$ with each $x_{i,1}, x_{i,2}, x_{i,3}$ being either an instance of a variable from $X$ or its negation with them being pairwise distinct. The idea is as follows. Each $\varphi_i$ will correspond to a gadget like the one depicted in Fig. \ref{gadgets}(a). The gadgets will be joined by edges in sequence, as seen in Fig. \ref{gadgets}(c), with constraints that restrict using nodes corresponding to an instance of a variable after using the node corresponding to the negation of said variable and vice versa to appear exactly once. It is very similar, albeit a little more complicated, to what was done in the proof of Theorem \ref{th_npc_maxonce}. In fact, it suffices to use $G=(V, E, s, t, C, C_1)$, a cCFG constructed with $V$, $E$, $s$, $t$, and $C$ identical to the one defined in the proof of Theorem \ref{th_npc_maxonce}, we will therefore omit writing the definition again.

%\begin{figure}[!ht]
%\centering
%\includegraphics[scale=0.2]{fig_singleGadget2}
%\caption{A single gadget representing a clause $\varphi_i$}
%\label{figGadget3}
%\end{figure}

Like in that proof, we will start with proving that if $\varphi$ is satisfiable, then there exists a test set $T_\varphi$ achieving EC for $G$ respecting $C$. Fix $\nu: X \rightarrow \{0, 1\}$, an assignment of logical values to the variables that satisfies $\varphi$. For simplicity we define $\overline\nu(x_j)$, $j \in \{1,\ldots,k\}$, as $x_j$ iff $x_j$ is set as true by $\nu$ and as $\neg x_j$ in the other case. For the first element of $T_\varphi$, choose a path $p$ starting with vertices $s$ and $b_1$, corresponding to $\nu$ in the following way. At each $v_i$, $i \in \{1, \ldots, n\}$, choose an edge leading to a vertex that corresponds to either a variable that is set as true by $\nu$ or a negation of a variable that is set as false by $\nu$ (one of these options must be possible at each $i$, since each $\varphi_i$ must be satisfied). Note that $p$ does not contain any pairs of vertices restricted by $C$.

Take any variable $x_j$, $j \in \{ 1,\ldots,k \}$, and construct a path $p_{x_j}$ from $p$ as follows. At each $v_i$, $i \in \{ 1,\ldots,n \}$, if one of the vertices among $v_{i,1}$, $v_{i,2}$, $v_{i,3}$ corresponds to $\overline\nu(x_j)$, make it a next vertex in $p_{x_j}$ instead of the one used by $p$ (due to elements of $\varphi_i$ being pairwise distinct, there is at most one such vertex). Note that $p_{x_j}$ also does not contain any pairs of vertices restricted by $C$, and also passes through every vertex in $G$ corresponding to $\overline\nu(x_j)$ (it may well be the case that there are no such vertices, it does not matter for the construction). Now, for each vertex $v_{\alpha,\beta}$ corresponding to $\neg\overline\nu(x_j)$, add to $T_\varphi$ a path constructed from $p_{x_j}$ by replacing (at most) two vertices in the following way. If there is an edge $e$ outgoing from $s$ that is unused by every path added to $T_\varphi$ so far, change the second vertex from the one used by $p_{x_j}$ to the vertex having $e$ as an incoming edge. Further, replace a vertex following $v_\alpha$ in $p_{x_j}$ by $v_{\alpha,\beta}$. Note that (since $p_{x_j}$ does not contain any pairs of vertices restricted by $C$), the only pairs from $C$ that are contained in the path after substituting $v_{\alpha,\beta}$ are of the form $(\cdot, v_{\alpha,\beta})$ or $(v_{\alpha,\beta},\cdot)$, with the other coordinate corresponding to $\overline\nu(x_j)$. What is more, all such pairs from $C$ are contained in $p_{x_j}$ after modifying it in this way (since $p_{x_j}$ contains all vertices corresponding to $\overline\nu(x_j)$).

%Note that after adding a path to $T_\varphi$ for every vertex corresponding to $\neg\overline\nu(x_j)$, there is now exactly one path for each pair in $C$ % \textcolor{blue}{(JR) ja bym troche pociagnal dowod w tym miejscu, jesli mamy $x_1$ odpowiednio w klauzulach $\varphi_1, \varphi_2, \varphi_3$, $x_1$ jest ustalone na true w wartosciowaniu $\nu$ oraz $x_1$ jest niezaprzeczone w $\varphi_1$ oraz zaprzeczone w $\varphi_2$ i $\varphi_3$, to ta powstala sciezka bedzie OK dla dwoch constraintow w $C$ - niby to jest nadal prawda, ale troche trzeba sie nad tym zastanowic, imo za szybko} containing a vertex corresponding to $x_j$, but paths constructed for $x_j$ do not contain pairs from $C$ containing any vertices corresponding to other variables.

Proceeding in an analogous way for every other variable from $X$, paths from $T_\varphi$ now respect $C$; it remains to ensure EC for $G$. Note that it follows from the construction that edges adjacent to each vertex $v_{l,m}$, $l\in \{1, \ldots, n\}$, $m\in \{1, 2, 3\}$ are already covered. Indeed, since each vertex $v_{l,m}$ corresponds to some variable $x_\iota$ or its negation, they were either covered by a path of the form $p_{x_\iota}$ if $\overline\nu(x_\iota)$ is true, or by one of its modifications if $\overline\nu(x_j)$ is false. 
The only edges that may not be covered are adjacent to some of the vertices from the set $\{b_1, \ldots, b_{(3n)^2+1}\}$. It suffices, therefore, to add to $T_\varphi$ paths constructed from $p$ by replacing a second vertex by one from $\{b_1, \ldots, b_{(3n)^2+1}\}$ unused so far. After doing so for each vertex that went unused, $T_\varphi$ achieves EC for $G$ and still respects $C$ (since $p$ does not contain any pairs of vertices restricted by $C$).

We will now prove that if there exists a test set $T_\varphi$ achieving EC for $G$ respecting $C$, then $\varphi$ is satisfiable. Since there are $(3n)^2+1$ edges outgoing from $s$, $T_\varphi$ must contain at least as many paths. Further, since there are $3n$ vertices that correspond to a variable or its negation (i.e. those of the form $v_{l,m}$, $l\in \{1, \ldots, n\}$, $m\in \{1, 2, 3\}$), there are at most $(3n)^2$ pairs in $C$. Therefore, since there are more paths in $T_\varphi$ than the number of pairs restricted by $C$ (with each pair appearing in exactly one path from $T_\varphi$), there must exist at least one path, let us denote it by $t_0$, which is in $T_\varphi$ but does not contain any pairs of vertices restricted by $C$. Denote the sequence of numbers from 1 to 3 that correspond to the path $t_0$ taken in each gadget corresponding to each $\varphi_i$ by $(k_1,\ldots,k_n)$. Define $\nu: X \rightarrow \{0, 1\}$ by setting values for the variables in $X$ such that each $x_{i,k_i}$ is true (note that it does not lead to contradictions thanks to $t_0$ not containing \textit{any} pairs of vertices restricted by $C$). For any remaining variable $x$, set $\nu(x)$ to any value. It is easy to see that $\nu$ satisfies $\varphi$.
\end{proof}

\subsection{ALWAYS}
 For the next result, we need a modification of the 3 SAT problem, namely POSITIVE 1-in-3 SAT (which is NP-complete \cite{DBLP:books/fm/GareyJ79}), that asks, for a given 3 SAT formula with only non-negated literals, if such a formula has a truth assignment such that each clause has exactly one true literal. We also make a technical assumption that each literal appears in at least one clause.

\begin{theorem}
\label{thm:complexity_always_relaxed}
Let $G=(V, E, s, t, C, C_=)$ be an acyclic cCFG with a set of ALWAYS constraints. Deciding whether a test set achieving EC for $G$ respecting $C$ of type $C_=$ exists is NP--complete.
\end{theorem}

\begin{proof}
Assume we have a formula $\varphi$ with clauses $\varphi_1, \ldots, \varphi_n$ and literals $X = \{x_1, \ldots, x_k\}$ with $|\varphi_i \cap \{x_1, \ldots, x_k \}| = 3$ . Denote $\varphi_i = x_{i,1} \lor x_{i,2} \lor x_{i,3}$. Let us consider $G=(V, E_{blue} \cup E_{red} \cup E_{black} \cup E_{green}, s, t, C, C_=)$, a cCFG such that:
\begin{itemize}
    \item $V=\{s, t, z_{n+1}, y_1,\ldots,y_k\} \cup \{v_{i, 1}, v_{i, 2}, v_{i, 3}, v_i, z_i, w_i, x_{i, t}, x_{i, f}\ |\ i \in \{1,\ldots,n\}\}$,
    \item $E_{blue} = \{(s,z_1), (z_{n+1}, y_1)\} \cup \{(w_1, z_2),(w_2, z_3) \ldots (w_n, z_{n+1}) \} \cup \{(z_1,v_1), \ldots, (z_n, v_n)\}$,
    \item $E_{red} = \{(s, v_i), (w_i, y_1)\ |\ i \in \{1,\ldots,n\}\},$
    %\item $E_2 = \{(s, v_i)\ |\ 1 \leq i \leq n \} \cup \{(w_i, y_1) |  1 \leq i \leq n \}$,
    \item $E_{black} = \{(v_j, v_{i,j}),  (v_{i,j}, w_i)\ |\ i \in \{1,\ldots,n\}, j \in \{1, 2, 3\}\}$
    \item $E_{green} = \{(y_i,x_{i,t}),(y_i,x_{i,f})\ |\ 1 \leq i \leq k\} \cup \{(x_{i,t}, y_{i+1}),(x_{i,f}, y_{i+1})\ |\ 1 \leq i < k \}\\$
    $\cup \{(x_{k,t}, t),(x_{k,f}, t)\}$ .
\end{itemize}

It remains to define a suitable set of constraints for $G$. Let $$C = Z \cup  \left(\bigcup_{1 \leq j \leq n} X_{\varphi_j}\right),$$ where 

\begin{itemize}
    \item $Z = \{(z_i,z_{i+1})\ |\ 1 \leq i \leq n \}$,
    %\item $X_{\varphi_j} = \{(v_{j,k}, x'_m)\ |\ k, m \in \{1, 2, 3\}\}$ for $x'_1, x'_2, x'_3 \in \varphi_j$ (% \textcolor{blue}{(JR) Wydaje mi sie, ze ta wersja jest niepoprawna - musi byc to co na dole})
    \item $X_{\varphi_j} = \{(v_{j,1}, x'_{1,t}),(v_{j,1}, x'_{2,f}),(v_{j,1}, x'_{3,f})\} \cup \{(v_{j,2}, x'_{1,f}),(v_{j,2}, x'_{2,t}),(v_{j,2}, x'_{3,f})\}$ \\$ \cup \{(v_{j,3}, x'_{1,f}),(v_{j,3}, x'_{2,f}),(v_{j,3}, x'_{3,t})\}$ for $x'_1, x'_2,x'_3 \in \varphi_j$.
\end{itemize}

Intuitively, constraints from $Z$ ensure that all gadgets corresponding to clauses must be on some path $p$ and constraints from $X_{\varphi_j}$ ensure that, if $x_i$ is chosen in some gadget corresponding to $\varphi_j$, then, if variables $x_i', x_i''$ are in $\varphi_j$, vertices $x_{i,t}, x_{i,f}', x_{i,f}''$ must be on $p$ and if $v_{j,i}$ (corresponding to a literal $x_m$) is on $p$, then in all clause gadgets that include the vertex $v_{j',i'}$ (corresponding to a literal $x_m$) is also on $p$. 

\begin{figure}[!ht]
\centering
\includegraphics[scale=0.48]{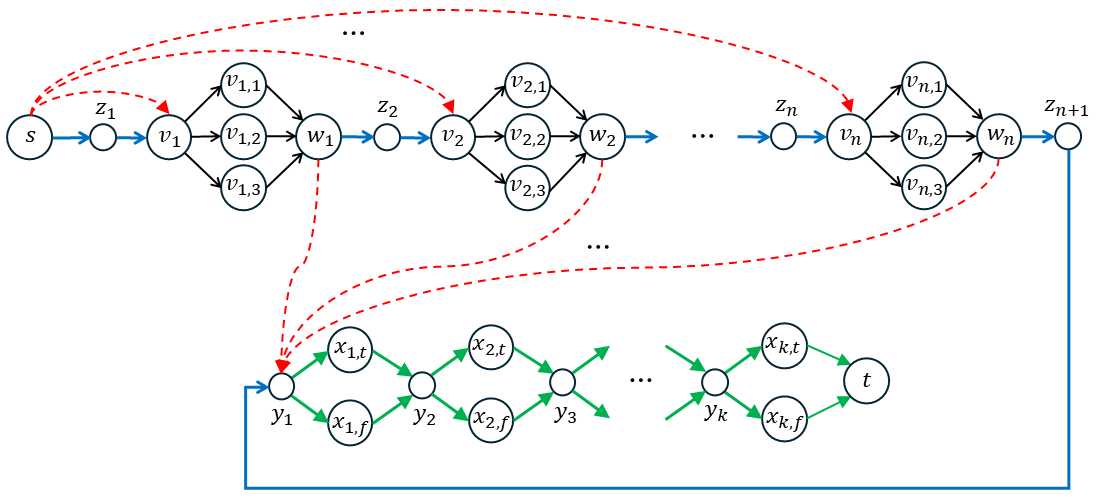}
\caption{Gadget used in Theorem \ref{thm:complexity_always_relaxed}}
\label{gadgets2}
\end{figure}

Obviously $G_\varphi$ is acyclic for every $\varphi$ (see Fig. \ref{gadgets2}).
We state and prove two claims:
\begin{lemma}
\label{lemma:cover_e2_e3_e4}
For every $\varphi$ there exists a set of paths $T$ from $s$ to $t$ in $G_{\varphi}$, respecting constraints from $C$, such that $E(T) = E_{red} \cup E_{black} \cup E_{green}$.
\end{lemma}

\begin{proof}
Consider clause $\varphi_i$ from $\varphi$ and let $x_j \in \varphi_i$. Assume that $x_j = x_{i,a}$ for some $a \in \{1,2,3\}$: 

\begin{itemize}
    \item begin $p_{i,a}$ with $(s,v_i, v_{i,a}, w_{i}, y_1)$ (edges from $E_2$ and $E_3$),
    \item append $(x_{1,f}, y_2, \ldots, x_{j-1,f}, y_{j}, x_{j,t}, y_{j+1}, x_{j+1,f}, \ldots, x_{n,f}, t)$ to $p_{i,a}$.
\end{itemize}

Obviously $p_{i,a}$ contains no $z_j$, so it needs only to respect constraints, $(v_{i,a}, x_{j,t})$, $(v_{i,a}, x'_{j,f})$ and $(v_{i,a}, x''_{j,f})$ for $x_j, x'_j, x''_j \in \varphi_i$. But, since $p_{i,a}$ contains only vertices corresponding to literal $x_j$ set to true (with the other two literals from the clause $\varphi_i$ set to false), all three constraints are fulfilled.

Let $T = \{p_{1,1}, p_{1,2}, p_{1,3}, \ldots, p_{n,1}, p_{n,2}, p_{n,3}\}$. Each $p_{i,a}$ corresponds to one literal of one clause of $\varphi$. It is easy to see, that all edges from $E_{black}$ are covered by $p_{i,a}$, since those are edges of form $(v_i, v_{i, a})$ or $(v_{i, a}, w_i)$, and each $x_{i,a}$ is covered by $p_{i,a}$. Further, edges from $E_{red}$ are also covered by $T$, since $(s,v_i)$ and $(w_i, y_1)$ are a part of $p_{i,a}$ for $i \in \{1, \ldots, n\}$ and $a \in \{1,2,3\}$. To show that all edges from $E_{green}$ are covered, observe that in each $p_{i,a}$ there is only one vertex $x_{j,t}$ on that path, so each $p_{i,a}$ covers edges $(a_1,x_{j,t})$ and $(x_{j,t},a_2)$ and $2n-2$ of form $(a_3,x_{j',f})$ and $(x_{j',f},a_4)$ and each $x_{j,t}$ must be at some $p_{i,a}$ -- otherwise there would be a literal that belongs to no clause. 
\end{proof}

\begin{lemma}
\label{lemma:assign_always_equiv}
The formula $\varphi$ has a truth assignment with exactly one true literal in every clause if and only if there exists a path $p$ in $G_{\varphi}$ from $s$ to $t$, respecting constraints from $C$, such that $E_{blue} \subset E(p)$.
\end{lemma}

\begin{proof}
Assume that $\varphi$ has a desired evaluation $\nu: \{x_1, \ldots, x_k\} \rightarrow \{0,1\}$. Let $\{x_{i_1}, \ldots , x_{i_l}\}$ be all literals, such that $\nu(x_{i_j}) = 1$. We construct desired path $p_{\nu}$ as follows:

\begin{itemize}
    \item $p_{\nu}$ starts with $s$, 
    \item for $j = 1$ to $n$ append to $p_{\nu}$ a path $(z_j, v_j, v_{j, a}, w_{j})$ such that $\nu(x_{j,a}) = 1$ (since $\nu$ is a truth assignment with one true literal per clause, then only one $x_{j,a}$ exists for every $\varphi_j$),
    \item append $(z_{m+1}, y_1)$ to $p_{\nu}$,
    \item append $(x_{1, b_1},y_2, \ldots, x_{k-1, b_{k-1}}, y_k, x_{k, b_{k}}, t)$ to $p_{\nu}$ where $b_i = t$ if and only if $\nu(x_i) = 1$.
\end{itemize}

Obviously, $p$ goes through all edges from $E_{blue}$ and fulfills all constraints of form $(z_j, z_{j+1})$. To show that constraints from $X_{\varphi_j}$ hold for every $1 \leq j \leq n$ observe that vertex $v_{j, a}$ that is on $p_\nu$ from clause $\varphi_j$ gadget imply three constraints from $X_{\varphi_j}$ that must be fulfilled - one corresponds to a true assignment of $x_{j, a}$ and the two others correspond to a false assignment of two other literals of $\varphi_j$. But, since $\nu$ is an evaluation that sets true to exactly one literal of every clause, every ``active'' condition of $X_{\varphi_j}$ holds for every $j$. 

Now, assume that there exists $p$ that covers $E_{blue}$ and respects constraints from $C$. Each path from $s$ to $t$ must contain $x_{i,t}$ or $x_{i,f}$, but not both of them at once, for all $i \in \{1, \ldots, k\}$, so we construct $\nu_p$ as follows:

\begin{itemize}
    \item set $\nu_p(x_i) = 1$ if  $p$ contains $x_{i,t}$,
    \item set $\nu_p(x_i) = 0$ if $p$ contains $x_{i,f}$.
\end{itemize}

Path $p$ must start with $(s, z_1)$. Since $p$ respects all constraints from $C$, then it hits all $n$ gadgets corresponding to clauses (otherwise it would be a pair $(z_i, z_{i+1})$ such that $z_i$ is on $p$ but $z_{i+1}$ is not on $p$). For the sake of contradiction, suppose that $p$ respects all constraints, but $\nu_p$ is not a truth assignment with exactly one true literal in each clause. Further consider two cases: 

Case 1 --  there exists $\varphi_j$ such that no literal from $\varphi_j$ is evaluated to true in $\nu_p$. Since $p$ hits all clause gadgets, it must also contain $v_{j, a}$. But then we have a constraint $(v_{j,a}, x_{i,t}) \in X_{\varphi_j}$ to fulfill, and $\nu_p(x_i) = 0$  -- so we have a contradiction, since $p$ does not respect all constraints.

Case 2 -- there exists $\varphi_j$ such that two or three literals are evaluated to true in $\nu_p$. Then there exist $x_i, x_{i'} \in \varphi_j$ such that $\nu_p(x_i) = \nu_p(x_{i'}) = 1$. But, since $p$ respects all constraints, it must also respect both $(v_{j,a}, x_{i,t})$ and $(v_{j,a}, x_{i,f})$ from $X_{\varphi_j}$. But that would mean that one of those conditions is violated in $p$, yielding a contradiction.  \end{proof}

To finish the proof of Theorem \ref{thm:complexity_always_relaxed} it suffices to observe, that every path that covers edge $(s,z_1)$ must also traverse through edges $(w_i,z_{i+1}), (z_i, v_i)$ and $(w_n, z_{n+1})$ (see constraints from $Z$). So we can cover $(s,z_1)$ with respect to constraints from $C$ if and only if there exists a path $p$ that goes through every $z_i$ at once. From the former and Lemmas \ref{lemma:cover_e2_e3_e4} and \ref{lemma:assign_always_equiv}, we deduce that Theorem \ref{thm:complexity_always_relaxed} holds.
\end{proof}

\section{FPT Algorithm for NEGATIVE}

We present here an algorithm that, for a given cCFG $G=(V, E, s, t, C, C_-)$ %% \textcolor{blue}{cCFG $G=(V, E, s, t, C_=, C_-)$ - naduzycie notacji, ale mozna na szybko poprawic w definicji jesli wszystkow tweirdzeniach sie zgadza}
,decides if there exists a test set that achieves EC for $G$ respecting $C$ of type NEGATIVE %% \textcolor{blue}{and POSITIVE mixed}
. The algorithm is Fixed Parameter Tractable (FPT) with respect to the number of constraints. 

Let $d = (d_1, ..., d_n)$ be a sequence. We write: $x \in d$ if $x = d_i$ for some $i \in \{1, \ldots, n\}$; $d^v = \{x: x \in d\}$; $\mathrm{last}(x, d) = \max\{i: d_i = x\}$ if $x \in d$ and 0 otherwise; $\mathrm{first}(x, d) = \min\{i: d_i = x\}$ if $x \in d$ and 0 otherwise. We also put $C^v = \bigcup_{(x,y) \in C} \{x,y\}$ %% \textcolor{blue}{$C^v = \bigcup_{(x,y) \in C_- \cup C_=} \{x,y\}$}
.

Let $C!$ be a set of all sequences $c=(c_1, \ldots, c_n)$ over $C^v$ in which every element from $C^v$ appears at most once. Let $I(x, c) = i$ if $c_i=x$, and 0 if $x \not \in c$.
Let $$\overline{C!} = \{c \in C! : \forall_{(x,y) \in C}\ (x \in c \wedge y \in c \wedge I(y,c) < I(x,c)) \lor x \not\in c \lor y  \not\in c\}.$$ 
% \textcolor{blue}{Let $\overline{C!} = \{c \in C! : \forall_{(x,y) \in C_{-}}\ (x \in c \wedge y \in c \wedge I(y,c) < I(x,c)) \lor x \not\in c \lor y  \not\in c \land \\ \forall_{(x,y) \in C_{=}}\ (x \in c \wedge y \in c \wedge I(x,c) < I(y,c)) \lor x  \not\in c\}.$}
For $c \in \overline{C}!$ and $p \in P(G)$, we say that $p$ is $c$-\textit{proper} if:

\begin{enumerate}
    \item \label{def:c_proper:cond1} $p^v \cap C^v = c^v$, 
    \item \label{def:c_proper:cond2} $I(v_1,c) < I(v_2,c) \Rightarrow \mathrm{last}(v_1, p) < \mathrm{last}(v_2, p)$,
    \item \label{def:c_proper:cond3} $(x,y) \in C$ % \textcolor{blue}{$(x,y) \in C_{-}$} 
    $\wedge x, y \in c \Rightarrow \mathrm{last}(y, p) < \mathrm{first}(x, p)$.
\end{enumerate}

Recall that $p \in P(G)$ respects a negative constraint $(x, y) \in C_-$ if $p \neq p_1.x.p_2.y.p_3$ for any $p_1, p_2, p_3 \in P(G)$.
% \textcolor{blue}{\begin{observation}\label{obs:cond1:always_path_resp}Let $p \in P(G)$ and $C_{=}$ be a set of ALWAYS constraints. The path $p$ respects all constraints from $C_{=}$ if and only if $\forall_{(x,y) \in C_{=}}\;x \notin p \lor (x \in p \land y\in p \land \mathrm{last}(x, p) < \mathrm{last}(y, p))$. \end{observation}\begin{proof}Chyba jest to dosc oczywiste jak sie zerknie na definicje tego warunku.\end{proof}}
Finally, for $c \in \overline{C!}$ and $y \in c$ define $\lambda_c(y) = \{ x \in c : (x,y) \in C\}$ % \textcolor{blue}{$\lambda_c(y) = \{ x \in c : (x,y) \in C_{-}\}$} 
and for $Y \subset c^v$ define $\lambda_c(Y) = \bigcup_{y \in Y} \lambda_c(y)$. We start with a simple observation.

\begin{observation}
    \label{obs:negative_path_resp}
    Let $G=(V, E, s, t, C, C_-)$ be a cCFG with a set of NEGATIVE constraints. A path $p$  respects all constraints from $C$ of type $C_-$ if and only if for every $(x,y) \in C$:
    % \textcolor{blue}{Let $p \in P(G)$ and $C_-$ be a set of NEGATIVE constraints. A path $p$  respects all constraints from $C_-$ if and only if for every $(x,y) \in C_-$:}
    
    \begin{enumerate}
        \item \label{obs:cond1:negative_path_resp} $x \notin p$ and $y \notin p$ or
        \item \label{obs:cond2:negative_path_resp}$x \in p$ and $y \notin p$ or 
        \item \label{obs:cond3:negative_path_resp}$x \notin p$ and $y \in p$ or 
        \item \label{obs:cond4:negative_path_resp}$p = p_1.y.p_2.x.p_3$ and $x \notin p_1$ and $x,y \notin p_2$ and $y\notin p_3$.
    \end{enumerate}
\end{observation}

\begin{lemma}
\label{lemma:constr_proper_equivalence}
    A path $p$ respects all constraints from $C$ of type $C_-$, if and only if there exists $c \in \overline{C}!$ such that $p$ is $c$-proper.
    % \textcolor{blue}{A path $p$ respects all constraints from $C_-$ and $C_=$, if and only if there exists $c \in \overline{C}!$ such that $p$ is $c$-proper.}
\end{lemma}
\begin{proof}
$(\Rightarrow)$ Assume that $p$ respects all constraints from $C$ of type $C_-$ % \textcolor{blue}{ from $C_-$ and $C_=$}
. Let $D = p^v \cap C^v$. We construct $c$ by ordering the vertices of $D$ according to their order of \emph{last occurrence} on $p$.
No vertex appears more than once in this ordering, hence $c \in C!$.
We now show that $c \in \overline{C!}$.
Consider any constraint $(x,y)\in C$ % \textcolor{blue}{ $(x,y)\in C_-$}
.
If both $x,y\in c$, then by Observation~\ref{obs:negative_path_resp}, $p$ must be of the form
$p=p_1.y.p_2.x.p_3$.
Hence $\mathrm{last}(y, p) < \mathrm{first}(x, p) \leq \mathrm{last}(x, p)$. Otherwise, $x\not\in c \lor y\not\in c$ holds. % \textcolor{blue}{Consider any constraint $(x,y) \in C_=$. If both $x,y \in c$, then by Observation \ref{obs:cond1:always_path_resp}, we know that Condition \ref{def:c_proper:cond2}  hold for $c$. Otherwise, $x \notin p$ and $y \in p$}. 
Thus $c\in\overline{C!}$. 
It remains to show that $p$ is $c$-proper.
By construction, $p$ contains no vertices from $C^v \setminus c^v$. Moreover, also by construction, whenever $I(v_1,c)<I(v_2,c)$, $\mathrm{last}(v_1, p) < \mathrm{last}(v_2, p)$. Finally, if there exists $(x,y) \in C$ % \textcolor{blue}{ $(x,y)\in C_-$} 
such that $x \in p$ and $y\in p$, then Condition \ref{obs:cond4:negative_path_resp} of
Observation~\ref{obs:negative_path_resp} must hold. So $x,y \in c$, but also $\mathrm{first}(x, p) > \mathrm{last}(y, p)$ (otherwise $p$ would not respect the $(x,y)$ constraint). Therefore, $p$ is $c$-proper.

$(\Leftarrow)$ Assume that there exists $c\in\overline{C!}$ such that $p$ is $c$-proper.
Consider any constraint $(x,y)\in C$ % \textcolor{blue}{$(x,y)\in C_-$}
.
\begin{itemize}
    \item If neither $x$ nor $y$ belongs to $c$, then neither appears on $p$, and the constraint is respected.
    \item If exactly one of ${x,y}$ belongs to $c$, then by $c$-properness the other does not occur on $p$, and the constraint is respected.
    \item If both $x,y\in c$, then since $c\in\overline{C!}$ we have $I(y,c)<I(x,c)$.
\end{itemize}
By $c$-properness, $\mathrm{last}(y, p) < \mathrm{first}(x, p)$, so $p$ is of the form $p=p_1.y.p_2.x.p_3$ for some $p_1, p_2, p_3 \in P(G)$ and $p$ satisfies Condition~\ref{obs:cond4:negative_path_resp} of Observation~\ref{obs:negative_path_resp}. In all cases, the constraint $(x,y)$ is respected by $p$. Since $(x,y)$ was arbitrary, $p$ respects all constraints from $C$ % \textcolor{blue}{$C_-$}
.
% \textcolor{blue}{
%Consider any constraint $(x,y)\in C_=$.
%\begin{itemize}
%    \item If $x$ does not belong to $c$, then it does not appear on $p$, and the constraint is respected.
%    \item If both $x,y\in c$, then since $c\in\overline{C!}$ we have $I(x,c)<I(y,c)$ but then, from Condition \ref{def:c_proper:cond2}, $\mathrm{last}(x, p) < \mathrm{last}(y, p)$ and, by Observation \ref{obs:cond1:always_path_resp}, the constraint is respected.
%\end{itemize}
%}
\end{proof}
We state another observation before presenting the algorithm. 
\begin{observation}
\label{obs:left_parts_of_constraints}
    Let $c \in \overline{C}!$ and $D = c^v \setminus \lambda_c(c^v)$. Then  $\forall y' \in D\;\forall (x,y') \in C % \textcolor{blue}{(x,y') \in C_-}
    : x \notin D$.
\end{observation}
\begin{algorithm}[!ht]
\caption{\textsc{FindProperPath}}
    \begin{algorithmic}[1]
    \Require cCFG $G=(V, E,s,t,C,C_-)$ % \textcolor{blue}{$G=(V, E,s,t,C_=,C_-)$}
    , $c \in \overline{C}!$,  $(u,v) \in E$
    \Ensure $c$-proper path $p=(s, ..., u, v, ..., t)$, or $\varnothing$ if such $p$ does not exist 
    \State \label{alg:find_proper_path:line_init_p} $j := 1, p := (s)$
    \State $D := c^v \setminus \lambda_c(c^v)$
    \State \label{alg:find_proper_path:line_g'_def} $G' := $ induced subgraph of $G$ on $V(G) \setminus (C^v \setminus D)$ 
        \While {$j \leq |c| \wedge p \neq \varnothing$} \label{alg:find_proper_path:while_start}
            \State $p' := \varnothing$
            \If{$p$ does not contain $(u,v)$ as an edge}
                \State $p' :=$ path in $G'$ from the last vertex of $p$ to $c_j$ that goes through $(u,v)$ 
            \EndIf
            \If{$p' = \varnothing$}
                \State $p' :=$ the path from the last vertex of $p$ to $c_j$ in $G'$ 
            \EndIf
            \If{$p' = \varnothing$}
                \State \label{alg:find_proper_path:invalid_end} \Return $\varnothing$ 
            \Else
                \State $G' :=$ induced subgraph of $G$ on $V(G') \setminus \{c_j\} \cup (\lambda_c(c_j) \setminus \lambda_c(V(G') \setminus \{c_j\}))$ \label{alg:maininvarian}
                \State $p := p.p'$
            \EndIf
            \State $j := j+1$
        \EndWhile 
        \If {$p$ does not contain $(u,v)$} \label{alg:find_proper_path:end_p_if}
            \If{there exists a path $p'$ from the last vertex of $p$ to $t$ in $G'$ through $(u,v)$}
            \State $p := p.p'$ \label{alg:adduv}
            \Else
            \State $p := \varnothing$
            \EndIf
        \Else
            \If{there exists a path $p'$ from the last vertex of $p$ to $t$ in $G'$}
            \State $p := p.p'$
            \Else
            \State $p := \varnothing$
            \EndIf
        \EndIf
    \State \label{alg:find_proper_path:valid_end} \Return $p$ 
    \end{algorithmic}
\end{algorithm}

\begin{lemma}
\label{lemma:find_proper_path_correct}
The procedure \textsc{FindProperPath}, for a given cCFG $G=(V, E,s,t,C,C_-)$ % \textcolor{blue}{$G=(V, E,s,t,C_=,C_-)$}
, $c \in \overline{C}!$, and $(u,v) \in E$ returns, in polynomial time, a $c$-proper path $p$ from $s$ to $t$ containing $(u,v)$ if such $p$ exists, or $\varnothing$ otherwise.
\end{lemma}
\begin{proof}
We start with the high-level idea of the proof. We are looking for a path $p$ that is $c$-proper and contains $(u, v)$ as an edge. We use an auxiliary induced subgraph $G'$ of $G$. During the main loop in line \ref{alg:find_proper_path:while_start}, all possible paths in $G'$ do not violate any constraint $(x, y)$ % \textcolor{blue}{$(x,y) \in C_-$}
with both endpoints $x$ and $y$ in $c$. From the fact that $c \in \overline{C!}$, if both $x, y$ are in $c$, then $y$ must precede $x$. When we use the right endpoint $y$ for some $(x, y) \in C$ for the last time in a path $p$ we are constructing, we ``unblock'' the possibility of using, in the remaining parts of $p$, left endpoints of constraints of the form $(x, y)$, but only if $x$ is not a part of $(x, y')$ % \textcolor{blue}{$(x,y') \in C_-$}
, where $y' \in V(G')$ (see line \ref{alg:find_proper_path:while_start}). The invariant of the ``while'' loop is that at the beginning of the $j$-th iteration, $c_i \not\in V(G')$ for $j > i$ and $c_j \in V(G')$. 

Notice that $c_1 \in V(G')$ at the beginning of the first iteration, because $c_1 \in D$. Suppose $c_1 \notin D$. Then there exists $1 < j \leq |c|$ such that $(c_1, c_j) \in C$ % \textcolor{blue}{$(c_1, c_j) \in C_-$}
, which contradicts the fact that $c \in \overline{C!}$.

For $j>1$, $c_j \in V(G')$ at the beginning of the $j$-th iteration. Suppose it is not the case. Then,  we would have $c=(... c_j ... x ... y ... )$ such that $(x,c_j) \in C$ % \textcolor{blue}{$(x,c_j) \in C_-$} 
and $(c_j, y) \in C$ % \textcolor{blue}{$(c_j, y) \in C_-$}
, so $c \notin \overline{C}!$.

Finally, notice that $c_i \not\in V(G')$ for $j>i$. Suppose in $j$-th iteration we add $c_i$ to $V(G')$. But we know that $c = (\ldots, c_i, \ldots, c_j, \ldots)$, so $(c_i, c_j)$ would belong to $C$ % \textcolor{blue}{$C_-$}, since we added it in the $j$-th iteration, and we have a contradiction.

%przez caly dowod w pomocniczym podgrafie G' grafu G bedziemy pilnowac, zeby bylo podgrafem G takim, ze pojscie po jakiejkolwiek sciezce w G' nie naruszalo zadnego constrainta, ktorego oba wierzcholki sa w c. Z definicji c-proper wynika, ze jesli oba wierzcholki jakiegos constrainta (lewy i prawy) sa w c, to w c prawy koniec musi byc zawsze wczesniej niz lewy. Jak po raz ostatni uzyjemy w sciezce wierzcholka ktory jest prawym koncem jakiegos constrainta, to odblokowujemy mozliwosc uzywania w dalszej konstrukcji sciezki lewego konca wszystkich constraintow, ktore jako prawy koniec mialy ten wierzcholek, chyba ze odblokowanie to spowodobaloby naruszenie innego constainta, ktorego prawy wierzcholek jest jeszcze w G'. Niezmiennikiem jest, ze po j-tym kroku w G' nie bedzie wierzcholkow od c1 do cj i na pewno bedzie cj+1.

The proof follows from Claims \ref{lemma:find_proper_path_correct:claim1}-\ref{lemma:find_proper_path_correct:claim3}.
    \begin{claim}
        \label{lemma:find_proper_path_correct:claim1}
        If $p$ returned by \textsc{FindProperPath} is not empty then:
        
        \begin{enumerate}
        \item \label{lemma:find_proper_path_correct:claim1:cond1} $p$ starts with $s$ and ends with $t$,
        \item \label{lemma:find_proper_path_correct:claim1:cond2} $p$ contains $(u,v)$,
        \item \label{lemma:find_proper_path_correct:claim1:cond3} $p$ is $c$-proper.
        \end{enumerate}
               
    \end{claim}
    \begin{proof}
        To prove Condition \ref{lemma:find_proper_path_correct:claim1:cond1} observe that if $p \neq \varnothing$, $p$ must start with $s$ (line \ref{alg:find_proper_path:line_init_p}), since the algorithm only appends new vertices to $p$ or sets $p$ to $\varnothing$. Obviously, the algorithm ends only in lines \ref{alg:find_proper_path:invalid_end} and \ref{alg:find_proper_path:valid_end} and if it does not return $\varnothing$ then the path must end with $t$.

        To prove Condition \ref{lemma:find_proper_path_correct:claim1:cond2} it suffices to notice that if the algorithm reaches the line \ref{alg:find_proper_path:end_p_if} then two cases may occur -- either the condition in line \ref{alg:find_proper_path:end_p_if} $(u, v)$ is false, so $p$ already contains $(u,v)$ -- or $p$ does not contain $(u,v)$. In that case it will be added in line \ref{alg:adduv}, hence the claim holds.

        To prove Condition \ref{lemma:find_proper_path_correct:claim1:cond3} we show that $p$ is $c$-proper.
    
        Property \ref{def:c_proper:cond1} ($p^v \cap C^v = c^v$). To show that $p^v \cap C^v = c^v$ 
       it is easy to check the invariant of the loop in line \ref{alg:find_proper_path:while_start}, that if the algorithm did not return $\varnothing$, then $p$ does contain $c_j$ after $j$-th iteration (in particular, before the first iteration of the loop, $c_1 \in V(G')$). Similarly $G'$ does not contain any vertices of $C^v$ besides those from $c^v$.

        Property \ref{def:c_proper:cond2} ($I(v_1,c) < I(v_2,c) \Rightarrow \mathrm{last}(v_1, p) < \mathrm{last}(v_2, p)$). Let $v_1 = c_j$ and $v_2=c_{j'}$ and $j<j'$. If the algorithm did not return $\varnothing$, then $p'$ does not contain $c_j$ after $j$-th iteration, since in $j$-th iteration it was removed from $G'$ and it cannot be added to $G'$ in further iterations. On the other hand, we know that in $j'$-th iteration $c_{j'}$ will be in $p'$.
        
        Property \ref{def:c_proper:cond3} ($(x,y) \in C$ % \textcolor{blue}{$(x,y) \in C_-$}
        $\wedge x, y \in c \Rightarrow \mathrm{last}(y, p) < \mathrm{first}(x, p)$). We show that if $p \neq \varnothing$ then $p$ respects the constraints from $C$ % \textcolor{blue}{$C_-$}
        . Consider graph $G'$ in line \ref{alg:find_proper_path:line_g'_def}. Observe that any vertex $y$ that is in $G'$ either is not a part of any constraint $c \in C$ % \textcolor{blue}{$c \in C_-$}
        , or from Observation \ref{obs:left_parts_of_constraints}, for every constraint $(x,y) \in C$ % \textcolor{blue}{$(x,y) \in C_-$}
        , holds that $x \notin V(G')$. Hence, every path in $G'$ respects trivially all constraints from $C$. Assume that $p$ is respects the constraints after $j$-th iteration. Consider $G'$ after $j$-th iteration. It is easy to see that it consists of previous edges of $G'$ without $c_j$ endpoint and those edges $(x, y) \in E(G)$ for which $(x,c_j) \in C$ % \textcolor{blue}{$(x,c_j) \in C_-$}
        and $\forall y' \in V(G')\ (x,y') \notin C$ % \textcolor{blue}{$(x,y') \notin C_-$}
        (see the definition of $\lambda_c$). We can deduce that $p.p'$ computed in the next iteration cannot violate any % \textcolor{blue}{NEGATIVE} 
        constraint.
    \end{proof}
    \begin{claim}
        \label{lemma:find_proper_path_correct:claim2}
        If there exists a path $p$ between $s$ and $t$, such that $p$ is $c$-proper and $p$ contains $(u,v)$, then the procedure \textsc{FindProperPath} does not return $\varnothing$, otherwise it returns $\varnothing$.
    \end{claim}
    \begin{proof}
        Assume that $(u,v)$ first occurrence in $p$ is between last occurrence of $c_j$ and $c_{j+1}$ for $j \in \{1, \ldots, |c|\}$ (analogous argument would be sufficient if $(u,v)$ would be between $s$ and $c_1$ or $c_{|c|}$ and $t$). Consider a subpath $p'$ of $p$ between last occurrences of $c_{j'}$ and $c_{j'+1}$ for $j' < j$. For every $i < j$ there cannot be $c_i$ on $p'$ and there cannot be any $c_{i'}$ such that $(c_{i'}, c_k) \in C$ % \textcolor{blue}{$(c_{i'}, c_k) \in C_-$}
        and $k > j$. But $c_{i'}$ and $c_i$ are the vertices that are not in $G'$ in $j'$-th iteration. The same argument holds if $j' > j$ and for $p'$ between $c_j$ and $c_{j+1}$. So in all cases we can find a path $p'$ and if any of $p'$ does not exist then $p$ cannot exist so the algorithm returns $\varnothing$.
    \end{proof}

    \begin{claim} \label{lemma:find_proper_path_correct:claim3}
        The procedure \textsc{FindProperPath} runs in polynomial time.
    \end{claim}
    \begin{proof}
        Obviously $|c| \leq 2|V|$. The ``while'' loop in line \ref{alg:find_proper_path:while_start} traverses $G'$ a constant number of times, inserts/deletes polynomial number of edges in $G'$ and adds/removes polynomial number of elements to/from the set. So the loop runs in polynomial time. The same holds for operations before and after the loop, so the result holds.
    \end{proof}
    The lemma holds by Claims         \ref{lemma:find_proper_path_correct:claim1}, \ref{lemma:find_proper_path_correct:claim2}
and \ref{lemma:find_proper_path_correct:claim3}.
\end{proof}

\begin{algorithm}[!ht]
\caption{\textsc{HasEdgeCoverageWithNegative}}
\label{alg:has_edge_cov}
\begin{algorithmic}[1]
\Require cCFG $G=(V, E, s, t, C, C_-)$ % \textcolor{blue}{$G=(V, E, s, t, C_=, C_-)$}
\Ensure True if there exist an admissible $TS_G$ that achieves EC for $G$, False otherwise
\State $T = \varnothing$
        \For{$c \in \overline{C!}$} 
            \State $E' = E(G) \setminus T$
                \While{$E'$ not empty} 
                    \State take $(u,v)$ from $E'$
                    \State $p = \textsc{FindProperPath}(G, c, (u,v))$
                    \If{$p \neq \varnothing$}
                        \State $T = T \cup\{p\}$
                        \State $E' = E' \setminus E(p)$
                    \Else
                        \State $E' = E' \setminus (u,v)$
                    \EndIf
                \EndWhile
            \EndFor
            \If{$T$ achieves EC for $G$}
                \State \Return True
            \Else
                \State \Return False
            \EndIf
\end{algorithmic}
\end{algorithm}

\begin{theorem}
    The algorithm \textsc{HasEdgeCoverageWithNegative} decides, in time $f(|C|)\cdot p(|V| + |E|)$, if there exists a test suite for $G$ that admits the set of NEGATIVE constraints $C$, where  $p(\cdot)$ is a polynomial function.
\end{theorem}

\begin{proof}
From Lemma \ref{lemma:constr_proper_equivalence} we know that every $p$ that admits constraint must be $c$-proper for some $c \in \overline{C}!$. So every path from $s$ to $t$ that covers $e \in E(G)$ must also be $c$-proper for some $c$. The algorithm checks for every $c \in \overline{C}!$ if $e \in E(G)$ can be covered by some $c$-proper path, so (by Lemma \ref{lemma:find_proper_path_correct}) it is correct. To show that it runs in desired time, observe that $|\overline{C}!| \leq |C!| = 2^{|C|!}$, and the procedure \textsc{FindProperPath} is executed at most $|\overline{C}!| \cdot |E|$ times, so the result holds.
\end{proof}

\section{Conclusions and further work}

Table \ref{tab:summary} summarizes our results achieved so far. They may be considered as initial results in the field. While all (except POSITIVE) types of constraints imply NP-completeness of the corresponding edge coverage problem, the question remains if there exists an FPT (with respect to the number of constraints) algorithm for ONCE, MAX-ONCE, and ALWAYS constraints. We strongly suspect that ALWAYS can adopt a similar approach to the NEGATIVE case, whereas the ONCE and MAX-ONCE cases require a different approach. 

Observe that in this paper, we considered only the computational complexity of the edge coverage problem for cCFG with only one type of constraints. A natural extension of the research would be to permit multiple types of constraints in a CFG, as this would better model real-world systems.

\begin{table}[h!]
\caption{Summary of the computational complexity results}
\centering
\begin{tabular}{rcc}
\hline
Constraint type & P/NPC & FPT? \\
\hline
POSITIVE ($C_+$) & P  & --  \\
NEGATIVE ($C_-$) & NPC & $\checkmark$  \\
ONCE ($C_1$) & NPC & ? \\
MAX-ONCE ($C_{\leq 1}$) & NPC & ? \\
ALWAYS ($C_=$) & NPC & ? % \textcolor{blue}{$\checkmark$}
\\
\hline
\end{tabular}

\label{tab:summary}
\end{table}

\bibliography{references}

\appendix

\end{document}